\documentclass[smallabstract,smallcaptions]{dccpaper}

\usepackage[T1]{fontenc}
\usepackage{epsfig}
\usepackage{amsmath}
\usepackage{amsthm}
\usepackage{amssymb}
\usepackage{color}
\usepackage{url}
\usepackage{algorithmicx}
\usepackage{algpseudocode}
\usepackage[dvipsnames]{xcolor}
\usepackage{wrapfig,lipsum,booktabs}
\usepackage[bf,small]{caption}
\usepackage{microtype}
\usepackage{setspace}

\newcommand{\myparagraph}[1]{\vspace*{-1.9ex}\paragraph*{\normalsize\bf{#1}}}
\def\hrulefill2{\leavevmode\leaders\hrule height 0.6pt\hfill}

\newlength{\figurewidth}
\newlength{\smallfigurewidth}
\algtext*{EndFunction}
\algtext*{EndProcedure}
\algtext*{EndFor}
\algtext*{EndWhile}
\algtext*{EndIf}
\algloopdefx{NIf}[1]{\textbf{if} #1 \textbf{then}}
\algloopdefx{NElse}{\textbf{else}}
\algloopdefx{NElseIf}{\textbf{else if}}
\algloopdefx{NForAll}[1]{\textbf{for each} #1 \textbf{do}}
\algloopdefx{NWhile}[1]{\textbf{while} #1 \textbf{do}}
\algloopdefx{NFor}[1]{\textbf{for} #1 \textbf{do}}

\newtheorem{lemma}{Lemma}
\newtheorem*{claim}{Claim}
\theoremstyle{definition}

\newcommand{\lvec}[1]{\overset{{}_{\leftarrow}}{#1}}
\newcommand{\lrange}[1]{\overleftarrow{#1}}
\newcommand{\clear}{\mathsf{rst}}
\newcommand{\hash}{\mathsf{hash}}
\newcommand{\lhash}{\mathsf{lhash}}
\newcommand{\nav}{\mathsf{nav}}
\newcommand{\nca}{\mathsf{nca}}
\newcommand{\lcp}{\mathit{lcp}}
\newcommand{\ISA}{\mathit{\lvec{SA}}}
\newcommand{\SA}{\mathit{SA}}
\newcommand{\hashArr}{\mathit{hs}}
\newcommand{\parent}{\mathit{par}}
\newcommand{\phrase}{\mathit{phr}}
\newcommand{\phash}{\mathit{hash}}
\newcommand{\lnk}{\mathit{lnk}}
\newcommand{\len}{\mathit{len}}
\newcommand{\str}{\mathit{str}}
\newcommand{\map}{\mathit{map}}
\newcommand{\chr}{\mathit{c}}

\begin{document}

\title
{\large\textbf{LZ-End Parsing in Compressed Space}}

\author{%
Dominik Kempa and Dmitry Kosolobov\\[0.5em]
{\small\begin{minipage}{\linewidth}\begin{center}
\begin{tabular}{ccc}
Department of Computer Science, \\
University of Helsinki, Finland \\
\url{dominik.kempa@cs.helsinki.fi}, \url{dkosolobov@mail.ru}
\end{tabular}
\end{center}\end{minipage}}
}

\maketitle
\thispagestyle{empty}

\vspace{0.5cm}
\begin{abstract}
We present an algorithm that constructs the LZ-End parsing (a variation of LZ77) of a given string of length $n$ in $O(n\log\ell)$ expected time and $O(z + \ell)$ space, where $z$ is the number of phrases in the parsing and $\ell$ is the length of the longest phrase. As an option, we can fix $\ell$ (e.g., to the size of RAM) thus obtaining a reasonable LZ-End approximation with the same functionality and the length of phrases restricted by $\ell$. This modified algorithm constructs the parsing in streaming fashion in one left to right pass on the input string w.h.p. and performs one right to left pass to verify the correctness of the result. Experimentally comparing this version to other LZ77-based analogs, we show that it is of practical interest.
\end{abstract}

\Section{Introduction}

The growth of the amount of highly compressible data in modern applications has accelerated the development of new compression algorithms working in space comparable to the size of their compressed input. The compression schemes based on the famous LZ77 algorithm~\cite{LZ77} have proved their extreme efficiency in compressing highly repetitive collections of genomes, logs, and repositories of version control systems. For such data, most other methods achieve significantly worse results. Unfortunately, the problem of the construction of LZ77-based schemes in small space and reasonable time is still very challenging (e.g., see \cite{BelazzouguiPuglisi,FischerEtAl,KarkkainenKempaPuglisi,Kosolobov,PolicritiPrezza} and references therein).

In this paper we consider a variant of LZ77 called LZ-End that was introduced in~\cite{KreftNavarro2, KreftNavarro}. This scheme is comparable in practice to LZ77 in the sense of compression quality (see~\cite{KreftNavarro}) but, in addition, allows to efficiently retrieve any substring of the compressed string (when equipped with an extra lightweight structure). The LZ-End construction algorithm presented in~\cite{KreftNavarro2} builds the LZ-End parsing of a string of length $n$ in $O(n)$ space, which is unacceptable for large inputs that do not fit in main memory. To our knowledge, there were no further improvements of this result.

We present an algorithm that constructs the LZ-End parsing of the input string of length $n$ in $O(n\log\ell)$ time w.h.p. (throughout the paper all logarithms have base~2) and $O(z + \ell)$ space, where $z$ is the number of phrases in the parsing and $\ell$ is an upper bound on the length of a phrase. Further, we modify this algorithm fixing $\ell$ in advance (e.g., to the size of main memory) and construct in $O(z + \ell)$ space an approximation of the LZ-End parsing in which all phrases have length less than $\ell$. We implement this version and experimentally show that it is of practical interest.

Recently, in~\cite{FischerEtAl} an algorithm was presented that constructs an approximation of LZ77 and possesses similar space and time characteristics. However, unlike the algorithm of~\cite{FischerEtAl}, ours does not require random access to the input and constructs the parsing in one left to right pass in expectation plus one right to left pass to verify that the parsing is correct, which is a good property in the external memory setting.

\myparagraph{Preliminaries.}
Let $s$ be a string of length $|s| = n$. We write $s[i]$ for the $i$th letter of $s$ and $s[i..j]$ for $s[i]s[i{+}1]\cdots s[j]$. The \emph{reversal of $s$} is the string $\lvec{s} = s[n]\cdots s[2]s[1]$. For any $i,j$, the set $\{k\in \mathbb{Z} \colon i \le k \le j\}$ is denoted by $[i..j]$; denote $[i..j) = [i..j] \setminus \{j\}$ and $(i..j] = [i..j] \setminus \{i\}$.
Our notation for arrays is similar: e.g., $a[i..j]$ denotes an array indexed by the numbers $[i..j]$. Let $h$ be a hash table mapping an integer set $S \subset \mathbb{N}$ into a set $T$. For $x \in \mathbb{N}$, denote by $h(x)$ the image of $x$ assuming $h(x) = \textbf{nil}$ if $x \notin S$.

The \emph{LZ-End parsing}~\cite{KreftNavarro} of a string $s$ is a decomposition $s = f_1 f_2 \cdots f_z$ constructed as follows: if we have already processed a prefix $s[1..k] = f_1 f_2 \cdots f_{i-1}$, then $f_i[1..|f_i|{-}1]$ is the longest prefix of $s[k{+}1..|s|-1]$ that is a suffix of a string $f_1 f_2 \cdots f_j$ for $j < i$; the substrings $f_i$ are called \emph{phrases}. For instance, the string $ababaaaaaac$ has the LZ-End parsing $a.b.aba.aa.aaac$. Then, the following lemma is straightforward.

\begin{lemma}
Let $f_1f_2\cdots f_z$ be the LZ-End parsing of a string. Then, for any $i \in [1..z)$, any proper prefix of length at least $|f_i|$ of the string $f_if_{i+1}\cdots f_z$ cannot be a suffix of a string $f_1f_2\cdots f_j$ for $j < i$.%
\label{TechLemma}
\end{lemma}

\vspace{-0.5cm}
\Section{Basic Observations}
\vspace{-0.2cm}

It is not immediately clear how to construct the LZ-End parsing due to its greedy nature. However, the definition of the LZ-End parsing easily implies the following observation suggesting a way how to perform the construction incrementally.

\begin{lemma}
Let $f_1 f_2 \cdots f_z$ be the LZ-End parsing of a string $s$. If $i$ is the maximal integer such that the string $f_{z-i}f_{z-i+1}\cdots f_z$ is a suffix of a string $f_1 f_2 \cdots f_j$ for $j < z - i$, then, for any letter $a$, the LZ-End parsing of the string $sa$ is $f'_1 f'_2 \cdots f'_{z'}$, where $z' = z-i$, $f'_1 = f_1, f'_2 = f_2, \ldots, f'_{z'-1} = f_{z'-1}$, and $f'_{z'} = f_{z-i}f_{z-i+1}\cdots f_z a$.%
\label{OnlineLZEnd}
\end{lemma}

It turns out, however, that the number of phrases that might ``unite'' into a new phrase when a letter has been appended (as in Lemma~\ref{OnlineLZEnd}) is severely restricted.

\begin{lemma}
If $f_1 f_2 \cdots f_z$ is the LZ-End parsing of a string $s$, then, for any letter $a$, the last phrase in the LZ-End parsing of the string $sa$ is 1)~$f_{z-1}f_za$ or 2)~$f_za$ or 3)~$a$.%
\label{Trihotomy}
\end{lemma}
\begin{proof}
By Lemma~\ref{OnlineLZEnd}, the last LZ-End phrase of $sa$ is $fa$, where $f = f_{z-i}f_{z-i+1}\cdots f_z$ for some $i$. Suppose, to the contrary, that $i > 1$. By the definition of LZ-End, there is $j < z - i$ such that $f$ is a suffix of $f_1f_2\cdots f_j$. If $|f_j| \le |f_{z-1}f_z|$, then $f$ has a proper prefix of length $|f| - |f_j| \ge |f_{z-i}|$ that is a suffix of $f_1 f_2 \cdots f_{j-1}$, which contradicts
Lemma~\ref{TechLemma}. If $|f_j| > |f_{z-1}f_z|$, then there is $j' < j$ (since $|f_j| > 1$) such that $f_j[1..|f_j|{-}1]$ is a suffix of $f_1f_2\cdots f_{j'}$ and, hence, the
prefix of length $|f_{z-1}f_z|-1$ of the string $f_{z-1}f_z$
is a suffix of $f_1 f_2 \cdots f_{j'}$, which
again contradicts Lemma~\ref{TechLemma}.
\end{proof}

Let $s$ be the input string of our algorithm and $n = |s|$.
The basic idea is
to read $s$ from left to right and compute the LZ-End parsing for each prefix of $s$ using a compressed trie storing all reversed prefixes of $s$ ending at the phrase boundaries of the current parsing: To extend the current prefix by a letter and rebuild the parsing, we check using the trie whether the last one or two phrases have
previous
occurrences ending at a phrase boundary; then, according to Lemmas~\ref{OnlineLZEnd} and~\ref{Trihotomy}, we unite zero, one, or two last phrases with the appended letter and thus obtain a new phrase.

This approach seems promising since
the trie can be stored in $O(z)$ space, where $z$ is the number of phrases in the current parsing. Unlike LZ77, however, the LZ-End parsing of a prefix of $s$ can have more phrases than the parsing of $s$ (e.g., $a.b.abb.ba.bb$ and $a.b.abb.babbc$). Nevertheless, the following lemma shows that the parsing of a prefix cannot have too many phrases.

\begin{lemma}
Denote by $z$ and $z'$, respectively, the numbers of phrases in the LZ-End parsing of strings $s$ and $s'$ such that $s'$ is a prefix of $s$. Then $3z \ge z'$.%
\label{PrefixPhraseNum}
\end{lemma}
\begin{proof}
Let $f_1f_2\cdots f_z$ and $f'_1f'_2\cdots f'_{z'}$ be the LZ-End parsings of $s$ and $s'$, respectively, and $z' > z$. Denote by $i$ the number such that $f_1 = f'_1, \ldots, f_i = f'_i$ and $f_{i+1} \ne f'_{i+1}$.

Obviously, the prefix of length $|f_1f_2\cdots f_{i+1}|$ of the string $s$ must have the parsing $f_1f_2\cdots f_if_{i+1}$. Further, it follows from Lemma~\ref{OnlineLZEnd} that the parsing of this prefix can be obtained from the parsing of $s'$ by an incremental process that appends letters to the right of $s'$ and, if necessary, unites one or two last phrases in the current parsing to produce a new phrase. Thus, since $f_{i+1} \ne f'_{i+1}$, we have $|f_{i+1}| \ge |f'_{i+1}f'_{i+2}\cdots f'_{z'}|$.

Now let us construct by induction a descending sequence $i_1 > i_2 > \ldots > i_k$ such that $k = \lfloor (z' - i)/2\rfloor$ and, for any $j \in [1..k]$, the string $f'_{i+j} f'_{i+j+1} \cdots f'_{z'-j+1}$ is a substring of the string $f_{i_j}$. Clearly, the existence of such sequence implies that $2z \ge 2k$ and, hence, $3z = 2z + z \ge 2k + (i + 1) \ge (z' - i - 1) + (i + 1) = z'$.

We put $i_1 = i$ as the base of induction. For the step of induction $j \in [2..k]$, assume that $i_1, i_2, \ldots, i_{j-1}$ is a sequence satisfying the induction hypothesis. By the definition of LZ-End, there is $i' < i_{j-1}$ such that $f_{i_{j-1}}[1..|f_{i_{j-1}}|{-}1]$ is a suffix of the string $f_1 f_2 \cdots f_{i'}$. Since, by the induction hypothesis, the string $f'_{i+j-1}f'_{i+j}\cdots f'_{z'-j+2}$ is a substring of $f_{i_{j-1}}$, the string $f'_{i+j-1}f'_{i+j}\cdots f'_{z'-j+1}$ must occur in $f_1 f_2 \cdots f_{i'}$; denote by $m$ the starting position of such occurrence. Denote by $i_j$ the minimal number such that $i_j \le i' < i_{j-1}$ and $m + |f'_{i+j-1} f'_{i+j} \cdots f'_{z'-j+1}| \le |f_1 f_2 \cdots f_{i_j}| + 1$. Now it suffices to show that $|f_1f_2 \cdots f_{i_j-1}| < m + |f'_{i+j-1}|$.

Suppose, to the contrary, that $|f_1f_2 \cdots f_{i_j-1}| \ge m + |f'_{i+j-1}|$. Then, the string $f'_{i+j-1}f'_{i+j}\cdots f'_{z'}$ has a proper prefix of length $|f_1f_2\cdots f_{i_j-1}| - m + 1 > |f'_{i+j-1}|$ that is a suffix of the string $f_1f_2\cdots f_{i_j - 1}$. This contradicts to Lemma~\ref{TechLemma}. So, $|f_1f_2 \cdots f_{i_j-1}| < m + |f'_{i+j-1}|$ and, hence, the string $f'_{i+j}f'_{i+j+1}\cdots f'_{z'-j+1}$ is a substring of $f_{i_j}$.
\end{proof}

For a string $t$, define as $\hash(t) = \sum_{i=1}^{|t|} t[i] \alpha^{i-1} \bmod \mu$ the \emph{Karp--Rabin fingerprint}~(e.g., see~\cite{CrochemoreRytter}) of $t$, where $\mu$ is a fixed prime such that $\mu \ge n^{c+4}$ for some $c \ge 1$, and $\alpha \in [0..\mu)$ is chosen uniformly at random during the initialization of the algorithm. Denote $\lhash(t) = \hash(\lvec{t})$. It is well known that the probability that two different substrings of $s$ have the same fingerprints is less than $\frac{1}{n^c}$; such situation is called a \emph{false positive}. Hereafter, we assume that there are no false positives to avoid repeating that the answers are correct with high probability. In the sequel we describe how to verify whether the constructed parsing really encodes the string $s$.

\Section{Fast Compressed Trie}
\vspace{-0.2cm}

Let $f_1 f_2 \cdots f_z$ be the LZ-End parsing of a prefix of $s$ that has just been calculated by our incremental algorithm. Our algorithm maintains a compressed trie $T$ containing the reversed prefixes $\lrange{f_1}$, $\lrange{f_1f_2}, \ldots,$ $\lrange{f_1f_2\cdots f_i}$ up to some specified index $i$. For each vertex $v$ of $T$, denote by $v.\parent$ the parent of $v$ (if any) and by $v.\str$ the string written on the path connecting the root and $v$ (note that $v.\parent$ and $v.\str$ are used only in discussions). Each vertex $v$ of $T$ contains the following fields: $v.\len$, the length of $v.\str$; $v.\map$, a hash table that, for any child $u$ of $v$, maps the letter $a = u.\str[v.\len{+}1]$ to $u = v.\map(a)$; $v.\phrase$, a number such that $v.\str$ is a prefix of the string $\lrange{f_1f_2 \cdots f_{v.\phrase}}$.

Define $\clear(x, i) = x\mathbin{\&}\neg(2^i - 1)$ (resetting $i$ least significant bits). For each non-root vertex $v$ in $T$, denote $p_v = \clear(v.\len, i)$ for the maximal $i$ such that $\clear(v.\len, i) > v.\parent.\len$ and denote $h_v = \hash(v.\str[1..p_v])$. For fast navigation in $T$, we maintain a hash table $\nav$ that, for each non-root vertex $v$, maps the pair $(p_v, h_v)$ to $v = \nav(p_v, h_v)$. The table $\nav$ allows us to parse the trie $T$ as follows (see Lemma~\ref{Navigation}):

\vspace{-0.1cm}
\noindent
\hrulefill2
\vspace{-0.15cm}

{\setstretch{0.95}
\begin{algorithmic}[1]
\footnotesize
\Function{$\mathsf{approxFind}$}{$pat$}
    \State $p \gets 0,\; v \gets root;$
    \For{$i \gets \lceil\log |pat|\rceil;\; i \ge 0;\; i \gets i - 1$}\label{lst:loopStart}
        \If{$v.\len \ge p + 2^i$}
            $ p \gets p + 2^i;$\label{lst:changeP}
        \ElsIf{$\nav(p + 2^i, \hash(pat[1 .. p + 2^i])) \ne \mathbf{nil}$}
            $ p \gets p + 2^i,$\label{lst:changeP2}
            $ v \gets \nav(p, \hash(pat[1 .. p]));$\label{lst:changeV}
        \EndIf\label{lst:loopEnd}
    \EndFor\vspace{-0.1cm}
    \If{$v.\map(pat[v.\len + 1]) \ne \mathbf{nil}$}
        $ v \gets v.\map(pat[v.\len + 1]);$\label{lst:lastStep}
    \EndIf\vspace{-0.1cm}
    \State\Return{$v$;}
\EndFunction
\end{algorithmic}
}

\vspace{-0.5cm}
\noindent
\hrulefill2
\vspace{0.1cm}

Our method resembles the so called fat binary search in z-tries introduced in~\cite{BelazzouguiBoldiPaughVigna} and the proof of its correctness in Lemma~\ref{Navigation} is essentially the same.

\begin{lemma}[see also \cite{BelazzouguiBoldiPaughVigna}]
Denote by $t$ the longest prefix of $pat$ that is represented in the trie~$T$. If $|t| = 0$, then $\mathsf{approxFind}(pat)$ returns the root of~$T$; otherwise, it returns a vertex $v$ such that $t$ is a prefix of $v.\str$ and either $v.\parent.\len < |t|$ or $v.\parent.\parent.\len < |t|$.%
\label{Navigation}
\end{lemma}
\begin{proof}
Since the case $|t| = 0$ is obvious, assume that $|t| > 0$. Denote by $v_0$ a vertex such that $t$ is a prefix of $v_0.\str$ and $v_0.\parent.\len < |t|$. Denote $v_1 = v_0.\parent$. It suffices to prove that $\mathsf{approxFind}$ returns either $v_0$ or a child of $v_0$. Suppose that $|t| < p_{v_0}$ (see Fig.~\ref{fig:approxfind}; the case $|t| \ge p_{v_0}$ is discussed in the end of the proof). We are to show that the loop in lines~\ref{lst:loopStart}--\ref{lst:loopEnd} finds $v = v_1$; then, obviously, the code in line~\ref{lst:lastStep} obtains $v = v_0$.

\begin{figure}[htb]
\includegraphics[scale=0.45]{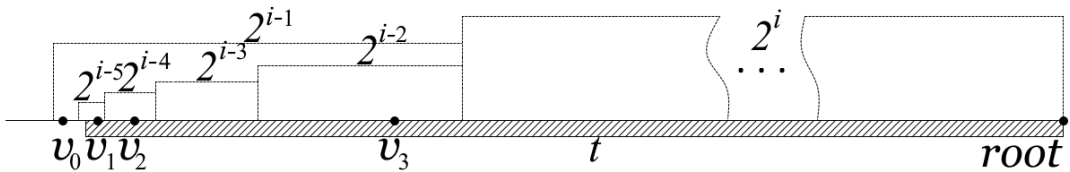}
\caption{Illustration for the proof of Lemma~\ref{Navigation}; here $p_{v_3} = 2^i, p_{v_2} = p_{v_3} + 2^{i-2}, p_{v_1} = p_{v_2} + 2^{i-3} + 2^{i-4}, p_{v_0} = p_{v_1} + 2^{i-5}$, and $|t| < p_{v_0}$.}\label{fig:approxfind}
\end{figure}

Suppose that the following invariants hold before each iteration of the loop~\ref{lst:loopStart}--\ref{lst:loopEnd}:\\
1. $v$ is a vertex lying on the path connecting the root and $v_1$;\\
2. either $v$ is the root and $p = 0$ or $v.\parent.\len < p \le v.\len$;\\
3. either $p = \clear(p_{v_1}, i + 1)$ or $v = v_1$.

Denote $j = \max\{j' \colon \clear(v_1.\len, j') = p_{v_1}\}$. By invariants~2--3 and definition of $p_{v_1}$, we have $p = p_{v_1}$ and $v = v_1$ before the iteration with $i = j - 1$. By invariant~1, all subsequent iterations do not change $v$ and, thus, the loop computes $v = v_1$.

It remains to prove that invariants~1--3 hold before each iteration. Since, initially, $v$ is the root and $p = 0 = \clear(p_{v_1}, \lceil\log |pat|\rceil + 1)$, invariants 1--3 hold before the first iteration. Invariant~2 is clearly preserved if $p$ is changed in line~\ref{lst:changeP}, and holds by the definition of $\nav$ if $p$ is changed in line~\ref{lst:changeP2}. Since $v$ is affected only by the code in line~\ref{lst:changeV}, invariant~1 is implied by the following straightforward claim and the definition of $v_1$.

\begin{claim}
For any $p' > p_{v_1}$, we have $\nav(p', \hash(pat[1..p'])) = \mathbf{nil}$.
\end{claim}

By the claim, $v$ cannot be changed once $v = v_1$; hence, invariant~3 is preserved if $v = v_1$. Thus, it remains to analyze invariant~3 when we have $v \ne v_1$.

Consider the case $p + 2^i \le v.\len$. By invariant~3, we have $p = \clear(p_{v_1}, i + 1)$. Then, the $i$th bit in $p_{v_1}$ must be equal to one because otherwise $p_{v_1} \le \clear(p_{v_1}, i + 1) + (2^i - 1) = p + 2^i - 1 < v.\len$, which contradicts to the inequality $p_{v_1} > v.\len$. Thus, the algorithm preserves invariant~3 assigning $p \gets p + 2^i$ in line~\ref{lst:changeP}.

Consider the case $p + 2^i > v.\len$.  By invariant~3, we have $p = \clear(p_{v_1}, i + 1)$. If $\clear(p_{v_1}, i) = \clear(p_{v_1}, i + 1)$, then we have $p + 2^i > p_{v_1}$ and, by Claim, $\nav(p + 2^i, pat[1 .. p{+}2^i]) = \mathbf{nil}$; hence, $p$ and $v$ remain unchanged as required. Finally, suppose that $\clear(p_{v_1}, i) = \clear(p_{v_1}, i + 1) + 2^i$. Denote $p' = \clear(p_{v_1}, i + 1) + 2^i$. Let $v'$ be a vertex on the path connecting the root and $v_1$ such that $v'.\parent.\len < p' \le v'.\len$; such $v'$ must exist because $p' \le p_{v_1}$. It follows from the assumption $p + 2^i > v.\len$ that $p \le v'.\parent.\len < p + 2^i \le v'.\len \le v_1.\len$. Thus, we have $p + 2^i = p_{v'}$ by the definition of $p_{v'}$ and, hence, $\nav(p + 2^i, \hash(pat[1 .. p + 2^i])) = v'$. According to this, the algorithm assigns $p \gets p + 2^i$ and $v \gets v'$ in line~\ref{lst:changeV} thus preserving invariant~3.

The case $|t| \ge p_{v_0}$ is similar: the loop~\ref{lst:loopStart}--\ref{lst:loopEnd} computes $v = v_0$ in the same way as it computes $v = v_1$ if $|t| < p_{v_0}$ but now $v$ may become a child of $v_0$ in line~\ref{lst:lastStep}.
\end{proof}

\Section{Algorithm}
\vspace{-0.2cm}

Let us first describe an algorithm with a parameter $\ell$ such that $\ell$ is an upper bound on the length of a phrase in the LZ-End parsing of the input string $s$. The algorithm scans $s$ from left to right and builds the LZ-End parsing for each prefix of $s$. We store the number of phrases in the current parsing in a variable $z$ and encode the parsing in an array $phrs[1..z]$ containing structures defined as follows: Suppose that $f_1 f_2 \cdots f_z$ is the parsing of the current prefix; then, for $i \in [1..z]$, we have $phrs[i].\chr = f_i[|f_i|]$, $phrs[i].\len = |f_i|$, $phrs[i].\phash = \lhash(f_i)$, and $phrs[i].\lnk$ is a number such that $f_i[1..|f_i|{-}1]$ is a suffix of $f_1 f_2 \cdots f_{phrs[i].\lnk}$ ($phrs[i].\lnk$ is arbitrary if $phrs[i].\len = 1$).

The algorithm reads $s$ by portions of length $\ell$; the processing of one portion is called a \emph{phase}. In the beginning of the $i$th phase ($i \ge 1$) $phrs[1..z]$ encodes the parsing of the string $s[1..i\ell{-}\ell]$ and the trie $T$ contains the reversed prefixes of $s$ ending at positions $\sum_{j=1}^{k} phrs[j].\len$ for all $k$ such that $\sum_{j=1}^{k-1} phrs[j].\len \le i\ell - 2\ell$. Since the length of any phrase is at most $\ell$, this guarantees that no prefix can be deleted from $T$ due to the changes in the array $phrs[1..z]$ during the future work of the algorithm.

\begin{lemma}[see~\cite{KreftNavarro}]
Suppose that the array $phrs$ encodes the LZ-End parsing $f_1 f_2 \cdots f_z$. Then, for any $j \in [1..z]$ and $k$, using $phrs$, one can retrieve the suffix of length $k$ of the string $f_1 f_2 \cdots f_j$ in $O(k)$ time.%
\label{StringRetrieval}
\end{lemma}

During the $i$th phase, we maintain integer arrays $lnks[i\ell - 2\ell .. i\ell]$ and $lens[i\ell - 2\ell .. i\ell]$ defined as follows. Let $m$ denote the length of the current prefix ($m = i\ell - \ell$ at the beginning of the phase). For each $j \in [\max\{1, i\ell - 2\ell\} .. i\ell]$, denote by $f_{j,1} f_{j,2} \cdots f_{j,z_j}$ the LZ-End parsing of $s[1 .. j]$. Then, for each $j \in [\max\{1, i\ell - 2\ell\} .. m]$, we have $lens[j] = |f_{j,z_j}|$ and the number $lnks[j]$ is such that $lnks[j] \in [1..z_j)$, $f_{j,z_j}[1..|f_{j,z_j}|{-}1]$ is a suffix of the string $f_{j,1} f_{j,2} \cdots f_{j,lnks[j]}$, and $\lrange{f_{j,1} f_{j,2} \cdots f_{j,lnks[j]}}$ is contained in the trie $T$, or we have $lnks[j] = \mathbf{nil}$ if there is no such number or $lens[j] = 1$ or $j \notin [1..m]$.

Define a function $\nca(z_1, z_2)$ that, for given $z_1$ and $z_2$, returns the nearest common ancestor of the leaves of $T$ corresponding to $\lrange{f_1 f_2 \cdots f_{z_1}}$ and $\lrange{f_1 f_2 \cdots f_{z_2}}$, where $f_1 f_2 \cdots f_z$ is the LZ-End parsing of the current prefix, or returns $\mathbf{nil}$ if one of these strings is not in $T$. We maintain on $T$ the structure of~\cite{ColeHariharan} that takes $O(z)$ space and can compute $\nca$ in $O(1)$ time using an array $N[1..z]$ such that $N[z']$ stores the leaf of $T$ corresponding to $\lrange{f_1 f_2 \cdots f_{z'}}$; $N$ is easily modified when a leaf is inserted or deleted. (Since all this $\nca$ machinery is quite complicated, in practice we use a simple naive solution, which appears to be very efficient.)

Denote $s_k{=}s[i\ell{-}3\ell..k]$. We begin the $i$th phase computing for the string $\lvec{s}_{i\ell}$ by standard algorithms (see~\cite{CrochemoreRytter}) the \emph{suffix array} $\SA$, its \emph{inverse} $\ISA$, and the array $\lcp[1..3\ell]$ that are defined as follows: $\SA[0..3\ell]$ is a permutation of $[i\ell - 3\ell .. i\ell]$ such that $\lvec{s}_{\SA[0]} < \lvec{s}_{\SA[1]} < \cdots < \lvec{s}_{\SA[3\ell]}$, $\ISA[i\ell - 3\ell .. i\ell]$ is such that $q = \ISA[\SA[q]]$, and for $q>0$, $\lcp[q]$ contains the length of the longest common prefix of $\lvec{s}_{\SA[q{-}1]}$ and $\lvec{s}_{\SA[q]}$. We equip $\lcp$ with the \emph{range minimum query (RMQ)} structure (e.g., see~\cite{CrochemoreRytter}) that uses $O(\ell)$ space and allows us to find the minimum in any range of $\lcp$ in $O(1)$ time. Then, we build an array $\hashArr[1..3\ell]$ such that $\hashArr[j] = \lhash(s_{i\ell - j})$. All this takes $O(\ell)$ time.

In addition, we maintain a balanced tree $P$ of size $O(\ell)$ that allows us to compute the maximum $\max\{k \in [1..z] \colon \sum_{j = 1}^k phrs[j].\len \le x\}$ for any $x \in [i\ell{-}3\ell .. i\ell]$ in $O(\log\ell)$ time. Finally, we construct a marked perfect binary tree $M$ with leaves $L[0..3\ell]$ in which a leaf $L[j]$ is marked iff a phrase of the current parsing ends at position $\SA[j]$, and an internal node of $M$ is marked iff it has a marked child (note that $M$ can be organized as an array of $O(\ell)$ bits).


\begin{table}
\caption{\label{tab:summary}\footnotesize
A summary of all described structures.}
\vspace{-0.25cm}
\small
\centering
{\setstretch{0.88}
\begin{tabular}{ r l }
\hline
fields of $phrs$ elem: & $phrs[j].\chr$, $phrs[j].\len$, $phrs[j].\phash$, $phrs[j].\lnk$ \\
fields of $T$ vertex:   & $v.\len$, $v.\map$, $v.\phrase$ \\\vspace{-0.1cm}
structures of $T$: & hash table $\nav(p, h)$, dynamic $\nca$ structure on $T$ \\
additional arrays: & $lnks$, $lens$, $\hashArr$, $\SA$, $\ISA$, $\lcp$, $N$ \\
miscellaneous:     & RMQ on $\lcp$, tree $P$, binary tree $M$ with leaves $L$ \\\hline
\end{tabular}
}
\vspace{-0.3cm}
\end{table}

The $i$th phase ($\mathsf{absorbTwo2}$, $\mathsf{absorbOne2}$, $\mathsf{updateRecent}$ are discussed below):

\noindent
\hrulefill2
\vspace{-0.15cm}

{\setstretch{0.95}
\begin{algorithmic}[1]
\footnotesize
\For{$m \gets i\ell - \ell + 1;\; m \le i\ell;\; m \gets m + 1$}\label{lst:phaseLoopBegin}\vspace{-0.06cm}
    \State $len \gets phrs[z].\len + phrs[z - 1].\len;$\vspace{-0.085cm}
    \State $p \gets \mathsf{approxFind}(\lrange{s[m - len .. m - 1]}).\phrase;$\label{lst:getP}
    \State $lnks[m] \gets \mathbf{nil};$\Comment{the global variable $ptr$ is set by $\mathsf{absorbOne2}$ and $\mathsf{absorbTwo2}$}
    \If{$len < \ell \mathop{\mathbf{and}} \mathsf{absorbTwo}(p, m)$}\label{lst:absorbTwoCheck}
        $z \gets z - 1,\; phrs[z].\len \gets len + 1,\;lnks[m] \gets p;$
    \ElsIf{$len < \ell \mathop{\mathbf{and}} \mathsf{absorbTwo2}(m)$}
        $ z \gets z - 1,\; phrs[z].\len \gets len + 1,\;p \gets ptr;$
    \ElsIf{$phrs[z].\len{<}\ell \mathop{\mathbf{and}} \mathsf{absorbOne}(p, m) $}\label{lst:absorbOneCheck}
        $ phrs[z].\len \gets phrs[z].\len{+}1,\;lnks[m] \gets p;$
    \ElsIf{$phrs[z].\len < \ell \mathop{\mathbf{and}} \mathsf{absorbOne2}(m)$}
        $ phrs[z].\len \gets phrs[z].\len + 1,\;p \gets ptr;$
    \Else
        $~z \gets z + 1,\;phrs[z].\len \gets 1;$
    \EndIf\vspace{-0.1cm}
    \State $lens[m] \gets phrs[z].\len;$
    \State $phrs[z].\chr \gets s[m],\;phrs[z].\phash \gets \lhash(s[m - phrs[z].\len + 1 .. m]),\;phrs[z].\lnk \gets p;$
    \State $\mathsf{updateRecent}();$\label{lst:phaseLoopEnd}
\EndFor\vspace{0.2cm}
\end{algorithmic}
\begin{algorithmic}[1]
\footnotesize
\Function{$\mathsf{absorbTwo}$}{$p, m$}
    \State\Return{$\mathsf{commonPart}(p, m, phrs[z].\len + phrs[z - 1].\len);$}
\EndFunction\vspace{0.2cm}
\end{algorithmic}
\begin{algorithmic}[1]
\footnotesize
\Function{$\mathsf{absorbOne}$}{$p, m$}
    \If{$phrs[p].\len < phrs[z].\len$}\label{lst:absorbOneCond}
        \Return{$\mathsf{commonPart}(p, m, phrs[z].\len);$}
    \EndIf\vspace{-0.1cm}
    \If{$phrs[p].c \ne phrs[z].c \;\mathop{\mathbf{or}}\; (phrs[z].len > 1 \mathop{\mathbf{and}} lnks[m - 1] = \mathbf{nil})$}\label{lst:absorbOneCheck2}
        \Return{$\mathbf{false};$}
    \EndIf\vspace{-0.1cm}
    \If{$phrs[z].len = 1$}
        \Return{$\mathbf{true};$}
    \EndIf\vspace{-0.1cm}
    \State\Return{$\nca(lnks[m - 1], phrs[p].\lnk).\len + 1 \ge phrs[z].\len$;}\label{lst:absorbOneCondBody}\vspace{-0.1cm}
\EndFunction\vspace{0.2cm}
\end{algorithmic}
\begin{algorithmic}[1]
\footnotesize
\Function{$\mathsf{commonPart}$}{$p, m, len$}
    \If{$phrs[p].\len \ge len \mathop{\mathbf{or}} phrs[p].\phash \ne \lhash(s[m - phrs[p].\len .. m - 1])$}\label{lst:absorbTwoCond}
        \Return{$\mathbf{false}$;}
    \EndIf\vspace{-0.1cm}
    \State $pos = m - phrs[p].\len;$
    \If{$lens[pos] - 1 + phrs[p].\len \ne len \mathop{\mathbf{or}} lnks[pos] = \mathbf{nil}$}\label{lst:lensAndLnks}
        \Return{$\mathbf{false};$}
    \EndIf\vspace{-0.1cm}
    \State\Return{$\nca(lnks[pos], p - 1).\len + phrs[p].\len \ge len;$}\label{lst:absorbTwoNCA}
\EndFunction
\end{algorithmic}
}

\vspace{-0.5cm}
\noindent
\hrulefill2
\vspace{0.2cm}

It is easy to see that we compute $\lhash$ (and no $\hash$) only for substrings of the string $s[i\ell{-}3\ell .. i\ell]$. It is well known that, using the array $\hashArr$ and the precomputed powers $\alpha^1, \alpha^2, \ldots, \alpha^{\ell}$ modulo $\mu$, one can compute $\lhash(s[j..j'])$ for any substring $s[j..j']$ of $s[i\ell{-}3\ell .. i\ell]$ in $O(1)$ time. Further, since the length of any phrase is less than $\ell$, we have $len \le 2\ell$ and, hence, we pass only reversed substrings of the string $s[i\ell{-}3\ell .. i\ell]$ to $\mathsf{approxFind}$. Therefore, the calculations of $\hash$ inside $\mathsf{approxFind}$ can also be performed in $O(1)$ time. Evidently, $\mathsf{approxFind}(pat)$ works in $O(\log|pat|)$ time. So, the phase processing works in overall $O(\ell\log\ell)$ time plus the time required for the functions $\mathsf{absorbTwo2}$, $\mathsf{absorbOne2}$, and $\mathsf{updateRecent}$ discussed below.

\begin{lemma}
Suppose that $f_1 f_2 \cdots f_z$ is the LZ-End parsing of $s[1..m{-}1]$ encoded in $phrs[1..z]$ at the beginning of an iteration of the loop~\ref{lst:phaseLoopBegin}--\ref{lst:phaseLoopEnd}. Then, the function $\mathsf{absorbTwo}$ [$\mathsf{absorbOne}$] in line~\ref{lst:absorbTwoCheck}~[\ref{lst:absorbOneCheck}] returns true iff the string $f_{z-1}f_z$ [$f_z$] is a suffix of a string $f_1 f_2 \cdots f_j$ whose corresponding reverse $\lrange{f_1 f_2 \cdots f_j}$ is in the trie $T$.%
\label{AbsorbFunctions}
\end{lemma}
\begin{proof}
Let $f_1 f_2 \cdots f_z$ be the LZ-End parsing of the string $s[1..m{-}1]$. Since we have $p = \mathsf{approxFind}(\lvec{f_z}\lvec{f}_{z-1}).\phrase$ (line~\ref{lst:getP}), it follows from Lemma~\ref{Navigation} that, if $f_{z-1}f_z$ is a suffix of a string $f_1 f_2 \cdots f_j$ whose corresponding reverse $\lrange{f_1 f_2 \cdots f_j}$ is contained in the trie $T$, then $f_{z-1}f_z$ must be a suffix of the string $f_1 f_2 \cdots f_p$.

Consider first the functions $\mathsf{absorbTwo}$ and $\mathsf{commonPart}$. Suppose that $|f_p| \ge |f_{z-1}f_z|$.
If $f_{z-1}f_z$ is a suffix of $f_1 f_2 \cdots f_p$, then $f_{z-1}f_z$ has a prefix of length $|f_{z-1}f_z|{-}1$ that is a suffix of $f_1 f_2 \cdots f_{phrs[p].\lnk}$, which is impossible by Lemma~\ref{TechLemma}.
The code in line~\ref{lst:absorbTwoCond} verifies the condition $|f_p| < |f_{z-1}f_z|$ and checks whether $f_p$ is a suffix of $f_{z-1}f_z$.

It remains to check whether the string $f' = s[m - |f_{z-1}f_z| .. m - |f_p| - 1]$ is a suffix of $f_1 f_2 \cdots f_{p-1}$. Obviously, the LZ-End parsing of the string $s[1..m{-}|f_{z-1}f_z|{-}1]$ is $f_1 f_2 \cdots f_{z-2}$. So, by the definition of LZ-End, if $f'$ is a suffix of $f_1 f_2 \cdots f_{p-1}$, then the string $s[1..m{-}|f_p|]$ has the parsing $f_1 f_2 \cdots f_{z-2} f''$, where $f'' = f' s[m{-}|f_p|]$; therefore, by the definition of $lens$ and $lnks$, we have $lens[pos] = |f''|$ and $lnks[pos] \ne \mathbf{nil}$, which is checked in line~\ref{lst:lensAndLnks}. Hence, $f'$ is a suffix of $f_1 f_2 \cdots f_{lnks[pos]}$. It is easy to see that the length of the longest common suffix of $f_1 f_2 \cdots f_{lnks[pos]}$ and $f_1 f_2 \cdots f_{p-1}$ is equal to $a.\len$, where $a$ is the nearest common ancestor of the corresponding leaves of $T$. So, we have $a.\len \ge |f'|$ iff $f'$ is a suffix of $f_1 f_2 \cdots f_{p-1}$, which is tested in line~\ref{lst:absorbTwoNCA}.

Consider the function $\mathsf{absorbOne}$. Since the case $|f_p| < |f_z|$ is analogous to the case $|f_p| < |f_{z-1}f_z|$ in $\mathsf{absorbTwo}$, we omit its analysis. Suppose that $|f_p| \ge |f_z|$ (lines~\ref{lst:absorbOneCond}--\ref{lst:absorbOneCondBody}). The case $|f_z| = 1$ is obvious, so, assume $|f_z| > 1$. Clearly, if $f_z$ is a suffix of $f_p$, then $f_z[1..|f_z|{-}1]$ is a suffix of the string $f_1 f_2 \cdots f_{phrs[p].\lnk}$. Hence, we have $f_z[|f_z|] = f_p[|f_p|]$ and, by the definition of $lnks$, $lnks[m{-}1] \ne \mathbf{nil}$. We check these conditions in line~\ref{lst:absorbOneCheck2}. Then, similar to the case $|f_p| < |f_z|$, we find the nearest common ancestor $a$ of the leaves of $T$ corresponding to $\lrange{f_1 f_2 \cdots f_{phrs[p].\lnk}}$ and $\lrange{f_1 f_2 \cdots f_{lnks[m{-}1]}}$ in line~\ref{lst:absorbOneCondBody} and, finally, have $a.\len{+}1 \ge |f_z|$ iff $f_z$ is a suffix of $f_p$.
\end{proof}

By Lemma~\ref{AbsorbFunctions}, the functions $\mathsf{absorbOne}$ and $\mathsf{absorbTwo}$ check whether the strings $f_z$ and $f_{z-1}f_z$ are suffixes of a prefix contained in $T$. But $f_z$ and $f_{z-1}f_z$ may have occurrences ending at the last position inside a phrase whose corresponding prefix does not belong to $T$. This case is processed by the functions $\mathsf{absorbTwo2}$ and $\mathsf{absorbOne2}$.

\noindent
\hrulefill2
\vspace{-0.15cm}

\noindent
\begin{minipage}[t]{0.5\textwidth}
\begin{algorithmic}[1]
\footnotesize
\Function{$\mathsf{absorbOne2}$}{$m$}
    \State\Return $\mathsf{chk}(m, phrs[z].\len);$
\EndFunction
\end{algorithmic}
\vspace{0.1cm}
\begin{algorithmic}[1]
\footnotesize
\Function{$\mathsf{absorbTwo2}$}{$m$}\vspace{-0.1cm}
    \State unmark leaf $L[\ISA[m - phrs[z].\len - 1]]$;
    \State $r \gets \mathsf{chk}(m, phrs[z - 1].\len + phrs[z].\len);$\vspace{-0.1cm}
    \State mark leaf $L[\ISA[m - phrs[z].\len - 1]]$;\vspace{-0.02cm}
    \State\Return{$r;$}
\EndFunction
\end{algorithmic}
\end{minipage}
\begin{minipage}[t]{0.5\textwidth}
\begin{algorithmic}[1]
\footnotesize
\Function{$\mathsf{chk}$}{$m, len$}
    \State $(ln, x) \gets \mathsf{markedLCP}(m - 1);$
    \If{$ln < len$}\label{lst:lnLen}
        \Return{$\mathbf{false}$;}
    \EndIf\vspace{-0.1cm}
    \State $ptr = \max\{k \colon \sum_{j = 1}^{k} phrs[j].\len \le x\};$\vspace{-0.02cm}
    \State\Return{$\mathbf{true};$}
\EndFunction
\end{algorithmic}
\end{minipage}
\vspace{0.15cm}
\begin{algorithmic}[1]
\footnotesize
\Function{$\mathsf{markedLCP}$}{$q$}\vspace{-0.05cm}
    \State $i' \gets \max\{i' \colon i' < \ISA[q]\text{ and }L[i']\text{ is marked}\}\text{ or }{+}\infty\text{ if there is no max};$\Comment{use $M$ here}\vspace{-0.1cm}
    \State $i'' \gets \min\{i'' \colon \ISA[q] < i''\text{ and }L[i'']\text{ is marked}\}\text{ or }{-}\infty\text{ if there is no min};$\Comment{use $M$ here}\vspace{-0.1cm}
    \State $y' \gets \min\{\lcp[j] \colon i' < j \le \ISA[q]\}\text{ or }0\text{ if }i' = {+}\infty;$\Comment{use RMQ here}\vspace{-0.1cm}
    \State $y'' \gets \min\{\lcp[j] \colon \ISA[q] < j \le i''\}\text{ or }0\text{ if }i'' = {-}\infty;$\Comment{use RMQ here}
    \If{$y' > y''$}
        \Return{$(y', \SA[i']);$}
    \Else
        ~\Return{$(y'', \SA[i'']);$}
    \EndIf
\EndFunction
\end{algorithmic}
\vspace{-0.5cm}
\hrulefill2

\begin{lemma}
Let $f_1 f_2 \cdots f_z$ be the LZ-End parsing of $s[1..m{-}1]$. For any $q \in [i\ell - 3\ell .. i\ell]$, $\mathsf{markedLCP}(q)$ finds in $O(\log\ell)$ time a pair $(ln, x)$ such that $L[\ISA[x]]$ is a marked leaf of $M$, $x \ne q$, $ln$ is the length of the longest common suffix of $s[i\ell - 3\ell .. q]$ and $s[i\ell - 3\ell .. x]$, and any other string $s[i\ell - 3\ell..p']$ such that $L[\ISA[p']]$ is marked and $p' \ne q$ has a shorter or the same longest common suffix with $s[i\ell - 3\ell .. q]$.%
\label{markedLCP}
\end{lemma}
\begin{proof}
The function uses the tree $M$ to find in $O(\log\ell)$ time the maximal number $i'$ (if any) and the minimal number $i''$ (if any) such that $i' < \ISA[q] < i''$ and the leaves $L[i']$ and $L[i'']$ of $M$ are marked. Then, using the RMQ data structure, we compute the minimums $y' = \min\{x \in \lcp[i'{+}1..\ISA[q]]\}$ and $y'' = \min\{x \in \lcp[\ISA[q]{+}1..i'']$\} (assuming that the minimum is~$0$ if the corresponding number $i'$ or $i''$ was not found) and, finally, find a position $x$ that is equal to either $\SA[i']$ or $\SA[i'']$ depending on the condition $y' > y''$. Obviously, $L[\ISA[x]]$ is marked. By standard arguments, one can show that the string $\lrange{s[i\ell - 3\ell .. x]}$ has the longest common suffix with the string $\lrange{s[i\ell - 3\ell .. q]}$ among all strings $\lrange{s[i\ell - 3\ell .. p']}$ such that $L[\ISA[p']]$ is marked and $p' \ne q$; moreover, the length of this suffix is $ln = \max\{y', y''\}$.
\end{proof}

Let $f_1 f_2 \cdots f_z$ be the LZ-End parsing of the string $s[1..m{-}1]$. By the definition of $M$, a leaf $L[\ISA[j]]$ is marked iff $j \in [i\ell - 3\ell .. i\ell]$ and $j = |f_1 f_2 \cdots f_k|$ for some $k \in [1..z]$. So, by Lemma~\ref{markedLCP}, if $f_1 f_2 \cdots f_z$ has a suffix of length $len \le \ell$ that is a suffix of a string $f_1 f_2 \cdots f_k$ such that $|f_1 f_2 \cdots f_k| \ge i\ell - 2\ell$, then we obtain $ln \ge len$ in line~\ref{lst:lnLen} in the function $\mathsf{chk}(m, len)$. In this case the function computes this number $k$ in $O(\log\ell)$ time using the tree $P$ and stores $k$ in the global variable $ptr$.

Thus, since the verification whether $f_z$ is a suffix of a string $f_1 f_2 \cdots f_k$ such that $|f_1 f_2 \cdots f_k| < i\ell - 2\ell$ is performed by $\mathsf{absorbOne}$, the call to $\mathsf{absorbOne2}(m)$ in the phase processing code returns true iff $f_z$ is a suffix of a string $f_1 f_2 \cdots f_k$ for $k \in [1..z)$ such that $|f_1 f_2 \cdots f_k| \ge i\ell - 2\ell$. Similarly, $\mathsf{absorbTwo2}(m)$ returns true iff $f_{z-1}f_z$ is a suffix of a string $f_1 f_2 \cdots f_k$ for $k \in [1..z{-}1)$ such that $|f_1 f_2 \cdots f_k| \ge i\ell - 2\ell$.

So, $\mathsf{absorbOne2}$ and $\mathsf{absorbTwo2}$ complement $\mathsf{absorbOne}$ and $\mathsf{absorbTwo}$ checking whether $\lvec{f}_z$ or $\lvec{f}_{z}\lvec{f}_{z-1}$ is a prefix of a string $\lrange{f_1 f_2 \cdots f_k}$ that is not contained in $T$. Thus, $lens[m]$ and $lnks[m]$ are filled with correct values. Finally, the function $\mathsf{updateRecent}$ performs in $O(\log\ell)$ time at most two unmarkings and one marking in the tree $M$ according to the updated array $phrs$, and modifies the tree $P$ appropriately.

\myparagraph{Phase postprocessing.} Once the $i$th phase is over, we must prepare all structures for the next phase. Let $f_1 f_2 \cdots f_z$ be the current parsing. First, we add to $T$ the strings $\lrange{f_1 f_2 \cdots f_{z'}}$, $\lrange{f_1 f_2 \cdots f_{z'+1}}, \ldots,$ $\lrange{f_1 f_2 \cdots f_{z''}}$, where $z'$ and $z''$ are such that $\lrange{f_1}, \lrange{f_1 f_2}, \ldots, \lrange{f_1 f_2 \cdots f_{z'-1}}$ are already in $T$, $\lrange{f_1 f_2 \cdots f_{z'}}$ is not in $T$, and $|f_1 f_2 \cdots f_{z''-1}| \le i\ell - \ell < |f_1 f_2 \cdots f_{z''}|$. The following lemma is an easy corollary of Lemma~\ref{TechLemma}.

\begin{lemma}
Let $f_1 f_2 \cdots f_z$ be the LZ-End parsing of a string. If the trie $T$ contains the strings $\lvec{f_1}, \lrange{f_1 f_2}, \ldots, \lrange{f_1 f_2 \cdots f_{j-1}}$ for $j < z$, then the longest prefix of the string $\lrange{f_1 f_2 \cdots f_j}$ that is represented in $T$ has length less than $|f_j|$.%
\label{LengthRestriction}
\end{lemma}

To insert $\lrange{f_1 f_2 \cdots f_j}$ (for $j = z', z'{+}1,\ldots,z''$) in $T$, we read $f_j$ right-to-left and traverse $T$ from the root like a Patricia trie using $v.\map$ in the traversed vertices $v$ and skipping the strings written on edges. Let $v$ be the deepest vertex found by this process. Then, we calculate the length of the longest common suffix of the strings $f_j$ and $f_1 f_2 \cdots f_{v.\phrase}$ by Lemma~\ref{StringRetrieval} (it is less than $|f_j|$ by Lemma~\ref{LengthRestriction}) thus obtaining the position in $T$ where the new leaf must be inserted. The $\nca$ data structure~\cite{ColeHariharan} is modified appropriately. (It is easy to see that Lemma~\ref{LengthRestriction} still holds in the presence of false positives; however, if $\ell$ artificially restricts the length of phrases and $|f_j| = \ell$, the longest common suffix of $f_j$ and $f_1 f_2 \cdots f_{v.\phrase}$ can be $f_j$ itself, but then we can ignore $f_j$ since the ``top $\ell$-part'' of $T$, which is actually important for us, remains correct.)

Denote by $u_0$ and $u_1$ the new leaf and its parent, respectively ($u_1$ might also be new). In an obvious way we calculate in $O(1)$ time the numbers $p_{u_0} = \clear(u_0.\len, k)$, for the maximal $k$ such that $\clear(u_0.\len, k) > u_1.\len$, and $h_{u_0} = \hash(\lrange{f_1 f_2 \cdots f_j}[1..p_{u_0}])$ using the array $\hashArr$; then, we assign $\nav(p_{u_0}, h_{u_0}) \gets u_0$.

Suppose that $u_1$ is a new vertex that has split the edge connecting a vertex $v$ and the old parent of $v$. As above, we calculate $p_{u_1}$ and $h_{u_1} = \hash(\lvec{f_{j}}[1..p_{u_1}])$, and assign $\nav(p_{u_1}, h_{u_1}) \gets u_1$. If $p'_v$, the old value of $p_v$, is greater than $u_1.\len$, then we are done. Suppose that $p'_v \le u_1.\len$. It follows from the definition of $p_{u_1}$ that in this case $p'_v = p_{u_1}$. Then, we recalculate $p_v$, compute $h_v = \hash(\lrange{f_1 f_2 \cdots f_{v.\phrase}}[1..p_v])$ in $O(p_v)$ time by Lemma~\ref{StringRetrieval}, and, finally, assign $\nav(p_v, h_v) \gets v$. By the definition of $p'_v$, we can have $p'_v \le u_1.\len$ only if $v.\len \le 2\cdot u_1.\len$, so, all this work takes $O(u_1.\len)$ time. Thus, the insertions altogether take $O(|f_{z'}| + |f_{z'+1}| + \cdots + |f_{z''}|) = O(\ell)$ time.

The new strings in $T$ require the rebuilding of $lnks$. First, we unmark in $O(\ell\log\ell)$ time all leaves $L[p]$ of $M$ such that $\SA[p] \ne |f_1 f_2 \cdots f_j|$ for any $j \in [z' .. z'']$. Then, for each $q \in [i\ell - \ell .. i\ell]$ such that $lnks[q] = \mathbf{nil}$, we compute $(ln, x) = \mathsf{markedLCP}(q - 1)$ and, if $lens[q] \le ln$, assign the number $\max\{k \colon \sum_{j = 1}^k phrs[j].\len \le pos\}$, which is computed by $P$ in $O(\log\ell)$ time, to $lnks[q]$. It follows from Lemma~\ref{markedLCP} and the bounding condition $lens[q] \le \ell$ that such algorithm indeed fills the array $lnks[i\ell - \ell .. i\ell]$ with correct values. Finally, we assign $i \gets i + 1$ and move to the next phase.

Thus, one phase including the postprocessing takes $O(\ell\log\ell)$ time and, therefore, the whole algorithm works in $O(n\log\ell)$ time and uses $O(z + \ell)$ space.

\myparagraph{The non-fixed $\ell$ and verification.}
We maintain a variable $\ell$ putting $\ell = 8$ initially and proceed as above. Once we obtain a phrase of length ${\ge}\frac{1}2 \ell$ during a phase processing, we put $\ell \gets 4\ell$ and start a new phase from this point rebuilding all internal phase structures; we also remove a number of leaves from the trie $T$ and modify the structures $\nav$, $\nca$, $N$ appropriately according to the phase processing of the above algorithm. Obviously, such algorithm works in $O(n\log\ell)$ overall time and constructs the LZ-End parsing with high probability. Note that this version is not streaming anymore since we reread a substring of length $2\ell$ each time the variable $\ell$ grows.

As we discussed above, the parsing is correct with high probability. To verify that possible false positives did not obscure the result, we read $s$ right-to-left and compare with the string retrieved from the parsing with the aid of Lemma~\ref{StringRetrieval}. If $\ell$ was fixed in advance and we intentionally did not produce phrases of length ${>}\ell$, then at this point we have a reasonable approximation of LZ-End that encodes the string $s$ and possesses properties similar to LZ-End. (We do not provide any theoretical evidence why this parsing is an approximation in a sense; we rather rely on intuition here.)

\Section{Experimental Results}
\vspace{-0.2cm}

\begin{wraptable}{r}{68mm}
\vspace{-3.2ex}
\caption{\label{tab:data}\footnotesize
Statistics of the testfiles used in experiments; $z$ is the
number of phrases in the LZ77 parsing, $z'$ is the number
of phrases in the LZ-End parsing with $\ell = 8 \times 2^{20}$.}
\vspace{-1ex}
{\setstretch{0.93}
  \footnotesize
  \begin{tabular}{l@{\hspace{1.6em}}r@{\hspace{1.6em}}r@{\hspace{1.6em}}r@{\hspace{1.6em}}r} \toprule
    Input  & $n/2^{30}$ & $\sigma$ & $n/z$ & $z'/z$ \\\midrule
    kernel & 128         & 229 & 4547.5 & 1.23       \\
    geo    & 128         & 211 & 3147.3 & 1.13       \\
    chr14  & 128         & 6   & 5957.9 & 1.25       \\ \bottomrule
  \end{tabular}
}
\vspace{-3ex}
\end{wraptable}
We implemented the algorithm described in this paper in C++ and compared
its runtime and the size of the resulting parsing to a number of LZ77
algorithms.
The experiments were performed on a machine equipped with two six-core 1.9 GHz
Intel Xeon E5-2420 CPUs with 15\,MiB L3 cache and 120\,GiB of DDR3 RAM. The
machine had 6.8\,TiB of free disk space striped with RAID0 across four
identical local disks achieving a transfer rate of ${\sim}$480\,MiB/s.
The OS was Ubuntu 12.04, 64bit running kernel 3.13.0. All programs were
compiled using {\tt g++} v5.2.1 with {\tt -O3} {\tt -DNDEBUG}
{\tt -march=native} options. The implementations of all algorithms used in
experiments are available at \url{https://www.cs.helsinki.fi/group/pads/}.
The experiments were run using three highly
repetitive testfiles  (see also Table~\ref{tab:data}):

\vspace{-1.5ex}
\begin{itemize}
  \setlength\itemsep{-1ex}
  \item[--] {\bf kernel}: a concatenation of source files from over 150
  versions of Linux kernel (\url{http://www.kernel.org/}); 
  \item[--] {\bf geo}: a concatenation of all versions (edit history) of
  Wikipedia articles about all countries and 10 largest cities in the XML
  format;
  \item[--] {\bf chr14}: multiple versions of \emph{Homo Sapiens} chromosome
  14 repeated to obtain a 128\,GiB file.
  Each version is obtained by randomly mutating the original chromosome with rate $0.01\%$. 
  See \url{http://hgdownload.cse.ucsc.edu/}.
\end{itemize}
\vspace{-1.5ex}

Text symbols are encoded using 8 bits and all algorithms in experiments use
40-bit integers to encode text positions.
The goal of our experiments is to determine: (1)~how scalable is the algorithm
described in this paper, and (2)~whether it is competitive with the best
external-memory algorithms computing LZ77 parsing.

The two fastest algorithm to compute the LZ77 parsing in external memory
are EM-LZscan and EM-LPF~\cite{KarkkainenKempaPuglisi2}.
EM-LZscan uses very little disk space and is very fast if the input is
highly repetitive and many phrases are entirely contained inside each other.
It gets slow, however, as the text-to-RAM ratio increases, since it needs to
scan essentially the whole text $n/M$ times, where $M$ is the size of available
RAM. EM-LPF, on the other hand, is more scalable, but since it needs the suffix
and LCP arrays as input, its disk space usage is at least 10 times the size of
the input text.

For experiments, we
fixed $\ell=8\times 2^{20}$, as it is small enough to not affect the RAM usage
significantly, and big enough to have essentially no effect on the parsing size.
In the preliminary run we executed our new algorithm on the full 128\,GiB instances of all three
testfiles, we recorded the following peak RAM usages: 4161\,MiB (kernel),
4557\,MiB (geo), and 3605\,MiB (chr14).

In the main experiment we executed all algorithms on increasing length prefixes
of all testfiles and measured the runtime. As explained above, for fair
comparison with the new algorithm, we allowed the LZ77 parsing algorithms to
use 3.5\,GiB of RAM (and we restricted the physical RAM available in the system
to 4\,GiB). After each run of the algorithm computing the LZ-End parsing, we
run the verifier on the resulting parsing (resulting in the second scan of the
input), but we never encountered any false positives. The time for the
verification is not included in the runtime of LZ-End parsing. The results are
given in Figure~\ref{fig:results}.

\begin{figure}[t]
  \centering
  \minipage{0.38\textwidth}
    \includegraphics[trim = 0mm 25mm 0mm 0mm, width=\linewidth]{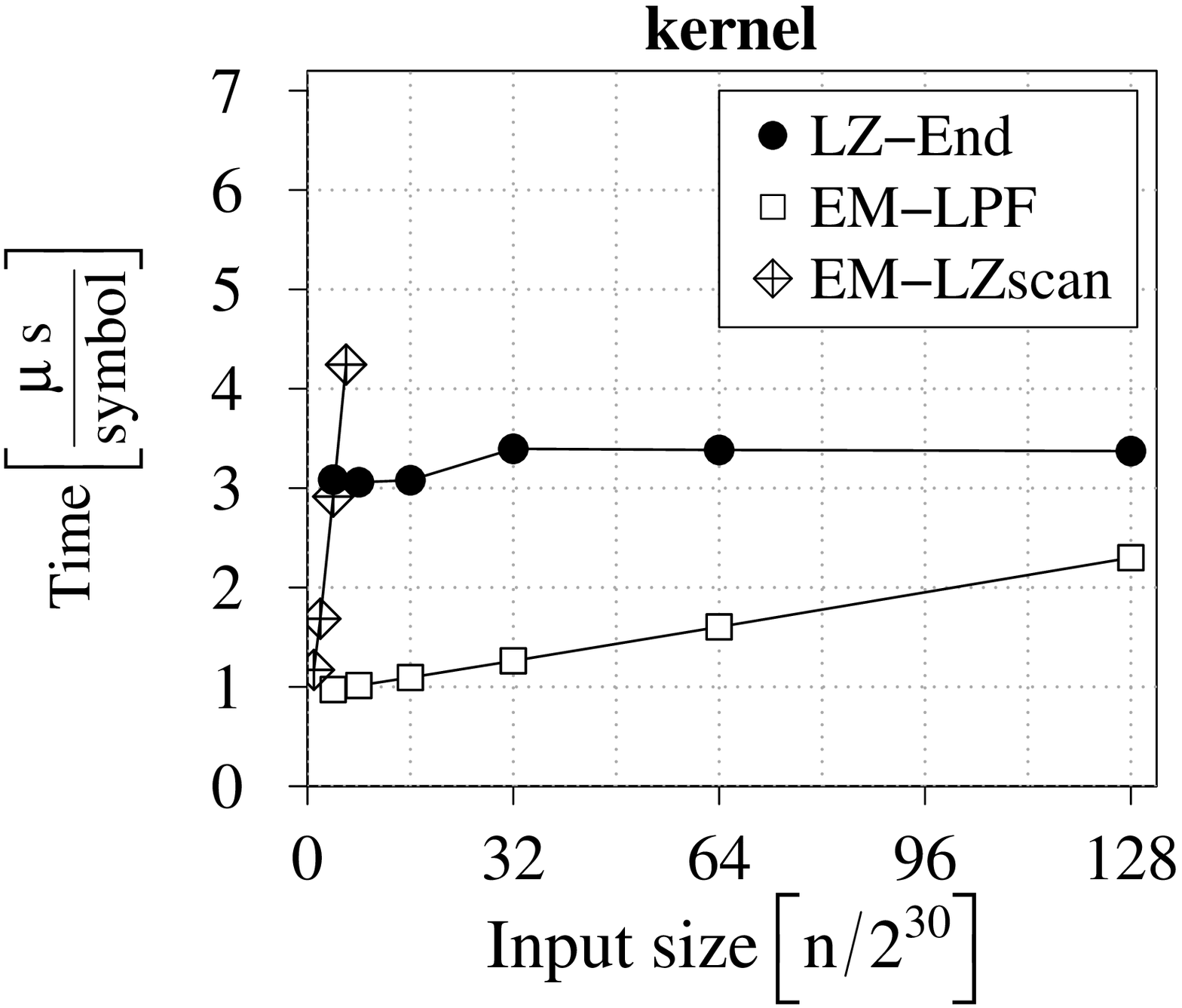}
  \endminipage
  \minipage{0.31\textwidth}
    \includegraphics[trim = 32mm 25mm 0mm 0mm, width=\linewidth]{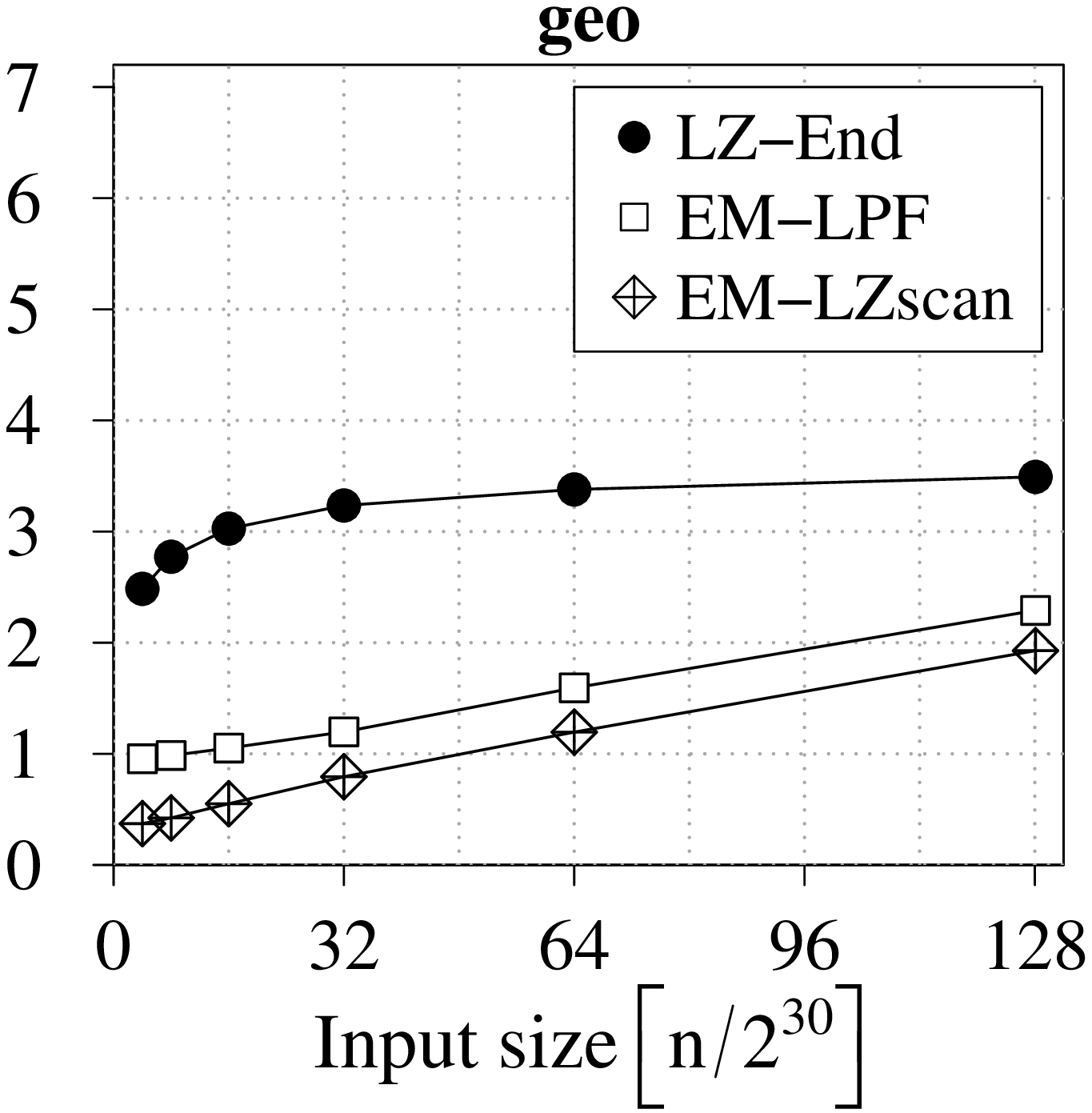}
  \endminipage
  \minipage{0.31\textwidth}
    \includegraphics[trim = 32mm 25mm 0mm 0mm, width=\linewidth]{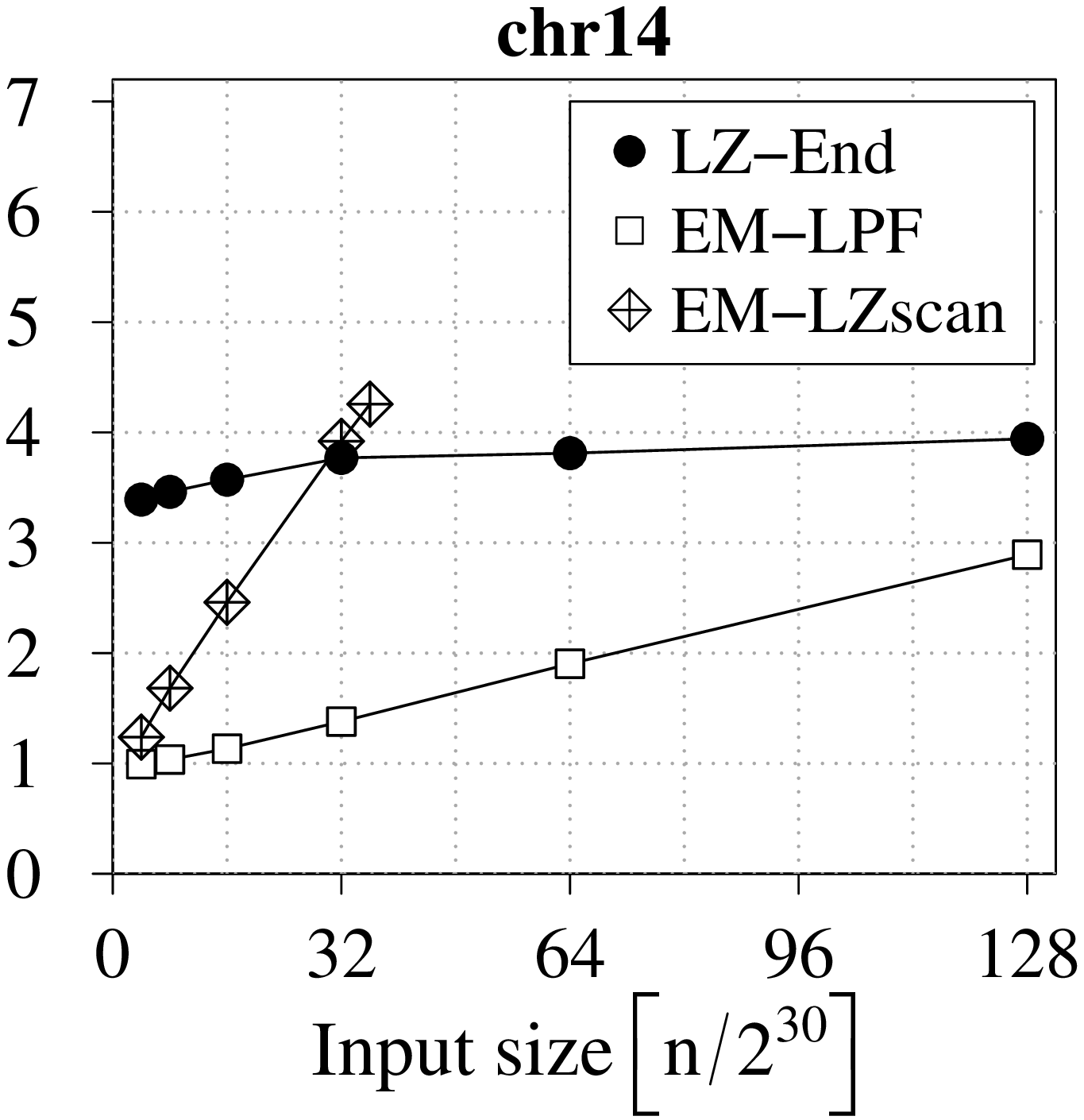}
  \endminipage
  \newline
  \vspace{3ex}
  \caption{\footnotesize Runtime (in $\mu s$ per input symbol) of the new algorithm compared to
    the fastest external-memory LZ77 parsing algorithms. EM-LPF and EM-LZscan use
    3.5\,GiB of RAM. EM-LPF includes the runtime for external-memory
    suffix~\cite{KarkkainenKempaPuglisi3} and LCP array construction~\cite{KarkkainenKempa}.}
  \label{fig:results}
  \vspace{-2ex}
\end{figure}

First, we observe that the algorithm to compute LZ-End scales very well with
increasing input. This is not surprising, as the algorithm has linear I/O
complexity. Second, the LZ-End construction is usually around two times
slower than EM-LPF, and up to four times slower than EM-LZscan, making our
LZ-End parser at least competitive with the existing LZ77 parsers.

It should be kept in mind, however, that because our LZ-End parser does not
need any disk space and only makes one left-to-right pass over the input (two,
if we include the verification) the algorithm has a number of properties that
none of the LZ77 algorithms have, e.g., the whole computation can be
performed over the network, or by decompressing the data on-the-fly. Our
algorithm only scans the input at a rate of 0.24--0.40\,MB/s which is well
below the typical network bandwidth, or the decompression speed of typical
modern decompressors like {\tt gzip} or {\tt bzip2}.
Lastly, we observe that the computed LZ-End parsing is never more than 25\%
larger than the size of LZ77 parsing (see Table~\ref{tab:data}), showing that
the LZ-End parsing is a valid replacement for LZ77 in practice.

\myparagraph{Acknowledgement.} The authors would like to thank Simon Puglisi and Juha K{\"a}rkk{\"a}inen for helpful discussions and comments that helped to improve the paper.

\makeatletter
\renewenvironment{thebibliography}[1]
     {\renewcommand{\baselinestretch}{\smallstretch}\footnotesize
      \setlength{\itemsep}{-0.155cm plus 0.3cm}
      \list{\@biblabel{\@arabic\c@enumiv}}%
           {\settowidth\labelwidth{\@biblabel{#1}}%
            \leftmargin\labelwidth
            \advance\leftmargin\labelsep
            \@openbib@code
            \usecounter{enumiv}%
            \let\p@enumiv\@empty
            \renewcommand\theenumiv{\@arabic\c@enumiv}}%
      \sloppy\clubpenalty4000\widowpenalty4000%
      \sfcode`\.\@m}
     {\def\@noitemerr
       {\@\LaTeX@warning{Empty `thebibliography' environment}}%
      \endlist}
\makeatother

\vspace{-0.2cm}
\paragraph{References}
\bibliographystyle{IEEEtranS}
\bibliography{refs}

\end{document}